\documentclass[runningheads]{llncs}
\usepackage{amsmath, amssymb}
\usepackage{graphicx}
\usepackage{enumitem}
\usepackage{xcolor}  
\usepackage{xspace}

\newcommand{\npc}{\textsf{NP}-complete\xspace}
\newcommand{\nph}{{\textsf{NP}\textrm{-hard}\xspace} }

\usepackage{comment}
\usepackage{graphicx}

\begin{document}
\title{Identifying Codes Kernelization Limitations}

\author{Aritra Banik\inst{1} \and
Praneet Kumar Patra \inst{2} \and
Adele Anna Rescigno\inst{3}\and Abhishek Sahu\inst{1}}
\authorrunning{F. Author et al.}

\institute{National Institute of Science, Education and Research, An OCC of Homi Bhabha National Institute\and
IISER Pune, India
\and
University of Salerno, Italy\\
}
\maketitle             

\begin{abstract}
The Identifying Code (IC) problem seeks a vertex subset whose intersection with every vertex’s closed neighborhood is unique, enabling fault detection in multiprocessor systems and practical uses in identity verification, environmental monitoring, and dynamic localization.  A closely related problem is the Locating-Dominating Set (LD), which requires each non-dominating vertex to be uniquely identified by its intersection with the set. Cappelle, Gomes, and Santos (2021) proved LD is W-hard for minimum clique cover and lacks polynomial kernels for parameters like vertex cover, but their methods did not apply to IC. This paper answers their question, showing IC does not admit a polynomial kernel parameterized by solution size plus vertex cover unless NP $\subseteq$ coNP/poly.

\end{abstract}


\section{Introduction}
Precisely identifying the location or state of every node within a network, using minimal resources, is a fundamental graph-theoretic problem. This challenge underpins diverse applications, from tracking mobile entities and ensuring robust system diagnostics to optimizing sensor placement in dynamic environments.
Consider, for instance, a large swarm of autonomous robots navigating a communication graph. Due to severe power constraints, these robots can only broadcast a localized signal to their immediate neighbors within the network infrastructure. A central server, tasked with continuous, precise localization of each robot, faces a critical question: \emph{What is the minimum number of stationary sensors required, and where should they be strategically placed in the network, to uniquely detect the location of every single robot?} This intrinsic and highly practical question directly leads to the \emph{Identifying Code} problem (IC). Formally, given a graph (or network) $G=(V,E)$, an identifying code is a subset of vertices $I \subseteq V$ such that for every vertex $v \in V$, its closed neighborhood $N[v]$ intersects $I$ in a unique, non-empty way. The objective of the Identifying Code problem is to find such a set $I$ with minimum cardinality.

Introduced by Karpovsky, Chakrabarty, and Levitin~\cite{karpovsky1998new} initially for fault diagnosis in multiprocessor systems, the Identifying Code problem swiftly found extensive applications across diverse domains. These include identity verification in distributed systems, efficient environmental monitoring, dynamic localization in wireless sensor networks~\cite{ray2003robust,ungrangsi2004implementation,ray2004robust},  biological applications \cite{HKS06}, and even the study of drug and terror social network analysis \cite{BS21}. For a comprehensive overview of its wide-ranging applicability, we refer the reader to the paper by Laifenfeld and Trachtenberg~\cite{laifenfeld2008identifying}.

Unfortunately, the Identifying Code problem has been proven to be \npc~\cite{charon2002identifying,charon2003minimizing}. A simple reduction from 3-SAT is presented in the well-known textbook by Cormen et al.~\cite{cormen2022introduction}. This inherent computational hardness naturally steers research towards approximation algorithms. Indeed, significant progress has been made, demonstrating that identifying codes can be approximated within an $\mathcal{O}(\log(|V(G)|))$ factor of the optimal solution~\cite{laifenfeld2006identifying,gravier2008hardness,suomela2007approximability}. However, this is essentially tight, as it has also been shown that the problem cannot be approximated in polynomial time within a factor of $1+\alpha \cdot \log(|V(G)|)$ for some constant $\alpha > 0$~\cite{laifenfeld2006identifying}. This inapproximability result suggests looking for solutions beyond approximation algorithms.

If a problem can be solved in time $f(k)\cdot \mathrm{poly}(n)$, where $f$ is a computable function depending only on the parameter $k$ and $n$ is the input size, then the problem is said to be \emph{fixed-parameter tractable (FPT)}. 
A closely related notion is that of a \emph{kernel}: a problem admits a kernel if there exists a polynomial-time algorithm that reduces any instance $(I,k)$ to an equivalent instance $(I',k')$ whose size is bounded by a function of $k$ alone.

Another closely related problem, the \emph{Locating-Dominating Set (LD)}, asks for a dominating set $D$ such that every vertex \emph{not} in $D$ is uniquely identified by its non-empty intersection with $D$. The parameterized complexity of the Locating-Dominating Set problem has recently garnered significant attention. Notably, Cappelle, Gomes, and Santos~\cite{CAPPELLE202168} demonstrated that LD is $\text{W[1]}$-hard when parameterized by the size of a minimum clique cover and, furthermore, that it does not admit polynomial kernels when parameterized by vertex cover or distance to clique. Despite these substantial results for LD, the authors explicitly highlighted a critical open challenge concerning Identifying Code, stating:
\begin{quote}
``As far as we could check, none of our approaches is applicable to Identifying Code, which may yield interesting complexity differences between these sibling problems.''
\end{quote}
In this paper, we directly address and resolve these fundamental open questions regarding the parameterized complexity of Identifying Code.

\section{No polynomial kernel for IC}

\def \ss {{Set Splitting }}
\def \F {{\cal F}}
\def \hF {{\hat \F}}
\def \bF {{\bar \F}}
\def \U {{U}}
\def \hd {{\hat{d}}}
\def \bd {{\bar d}}
\def \a {{\lambda}}
\def \b {{\omega}}

Given a collection $\F$ of subsets of a finite set $\U$, the {\em \ss} problem asks to find a partition
of $U$ into two disjoint subsets $\U_0$ and $\U_1$ which splits all the subsets in $\F$, i.e. no subset in $\F$ is entirely contained in either $\U_0$ or $\U_1$.
It is known that the \ss problem is an NP-hard problem \cite{garey1979computers}.

Bodlaender et. al. proved the well-known result below.
\begin{theorem}\label{th:BJK}
\cite{BodlaenderJK14}
If an \nph problem $P$ OR-cross-composes into the parameterized problem $Q$,
then $Q$ does not admit a  polynomial kernelization unless
$NP \subseteq coNP/poly$.
\end{theorem}

In the following, we prove that the \ss problem OR-cross-compose to the \textsc{Identifying Code} problem. By using Theorem \ref{th:BJK} we then have our main result.
\begin{theorem}\label{th:no-k}
The \textsc{Identifying Code} problem does not admit a polynomial kernel when parameterized by the size of the solution unless $NP \subseteq coNP/poly$.
\end{theorem}
Given a finite set $\U=\{x_1, \ldots, x_n\}$, we show how to encode $t$ instances  $(\F_0,\U), \ldots, (\F_{t-1},\U)$  of \ss problem into a single instance $\langle G,k \rangle$ of \textsc{Identifying Code} problem.
By copying instances, we may assume without loss of generality that $t=2^{s}$
, for some $s \in \mathbb{N}$.
Hence, we represent any index $0 \leq i \leq t-1$  through a binary string $i=0b_{s-1}\ldots b_1b_0$ (i.e., we add a leading bit  $b_{s} =0$ and $b_j \in \{0,1\}$ for $0 \leq j \leq s-1$).
Before constructing the graph $G$, for each input instance $(\F_i, U)$, 
 we consider
a new family $\hF_i$ obtained by including in it the following two sets, for each set $A
\in \F_i$:
    \begin{equation}\label{Aadded}
    A \cup \{d^{b_j}_j : 0 \leq j \leq s\} \ \ \text{ and }  \ 
    \ A \cup \{d^{1-b_j}_j : 0 \leq j \leq s\}.
    \end{equation}
We get $|\hF_i|=2|\F_i|$. The graph $G$ is constructed as  

\begin{figure}
        \centering
        \includegraphics[width=1\linewidth,keepaspectratio]{
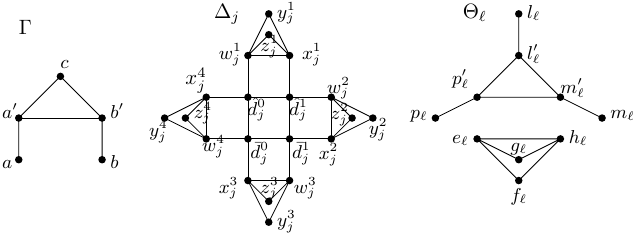}
        \caption{Gadgets used in the construction of graph $G$}
        \label{fig:gadgets}
    \end{figure}
    
follows:\\
$\bullet$ \  Create two independent sets $\U^+$ and $\U^-$ of $n$ vertices, each representing the $n$ elements in $\U$; i.e.,
$\U^+=\{x_\ell^+ \ | \ x_\ell \in \U\}$ and $\U^-=\{x_\ell^- \ | \ x_\ell \in \U\}$. \\
$\bullet$ \  Create three independent sets $\bF_*$, $\bF_-$ and $\bF_+$ of $\sum_{i=1}^t |\hF_i|$ vertices, each representing the elements in $\hF_i$ for each $i \in \{0, \ldots, t-1\}$; i.e., $\bF_* = \{v_B^* \ | \ B \in \hF_i, \ 0 \leq i \leq t-1 \}$, $\bF_+ = \{v_B^+ \ | \ B \in \hF_i, \ 0 \leq i \leq t-1 \}$ and $\bF_- = \{v_B^- \ | \ B \in \hF_i, \ 0 \leq i \leq t-1 \}$.
For each  $B \in \hF_i$ and $i \in \{0,\ldots,t-1\}$, if $x_\ell \in B$ then
connect vertex $v_B^* \in \bF_*$ to vertices  $x_\ell^+$ and $x_\ell^-$, 
connect vertex $v_B^+ \in \bF_+$ to vertex $x_\ell^+$, and finally
connect vertex $v_B^- \in \bF_-$ to vertex $x_\ell^-$.
\\
$\bullet$ \  Create an {\em instances distinguishing gadget} $\Lambda$. It consists of $2(s+1)$ vertices, $\a_0, \ldots, \a_{s}$ and $p_{\a_0}, \ldots, p_{\a_{s}}$, where each vertex $\a_j$ is connected to $p_{\a_j}$ for $j \in \{0, \ldots, s\}$. To implement the distinction of an instance, the vertices in $\Lambda$ are connected as follows to the vertices in the sets $\bF_*, \bF_-, \bF_+$:
For each $B \in \hF_i$, connect vertices $v_B^* \in \bF_*$, $v_B^- \in \bF_-$ and $v_B^+ \in \bF_+$ to each vertex $\a_j$ such that $b_j=1$ where $i=0b_{s-1}\ldots b_1b_0$ is the binary representation of $i$. We then make $\a_s$ adjacent to all the elements in $\bF_* \cup\bF_+\cup \bF_-$.
Notice that for $i \neq i'$, the vertices  $v_B^+$ (resp., $v_B^-, v_B^*$) and $v_{B'}^+$ (resp., $v_{B'}^-, v_{B'}^*$) for $B \in \hF_i$ and $B' \in \hF_{i'}$ have at least a different $\a$ neighbor in $\Lambda$. 
\\
$\bullet$ \ 
Let $\hF_i= \{B_{i,0}, \ldots, B_{i,m_i} \}$ and
let $m=\max_{i \, : \, 0\leq i\leq t-1}m_i$. We need $r=\lfloor \log_2 m \rfloor +1$ bits to have the binary representation of any $q$ such that $B_{i,q} \in \hF_i$.
 Create a {\em sets distinguishing gadget} $\Omega$.
It consists of $2r$ vertices, $\b_0, \ldots, \b_{r-1}$ and $p_{\b_0}, \ldots, p_{\b_{r-1}}$, where each vertex $\b_k$ is connected to $p_{\b_k}$ for $k \in \{0, \ldots, r-1\}$.  
To implement the distinction among the sets of an instance, the vertices in $\Omega$ are connected as follows to the vertices in the sets $\bF_*, \bF_-, \bF_+$:
connect vertex $\b_k$ to $v_{B_{i,q}}^* \in \bF_*$,  $v_{B_{i,q}}^+ \in \bF_+$ and  $v_{B_{i,q}}^- \in \bF_-$ whenever the bit $k$ in the binary representation of $q$ is 1.
Notice that fixed $\hF_i$, for $q \neq q'$, the vertices  $v_{B_{i,q}}^+$ (resp., $v_{B_{i,q}}^-, v_{B_{i,q}}^*$) and $v_{B_{i,q'}}^+$ (resp., $v_{B_{i,q'}}^-, v_{B_{i,q'}}^*$)  have at least a different neighbor in $\Omega$. 
\\
$\bullet$ \  Create the gadgets $\Delta_j$ for each $j \in \{0, \ldots, s\}$ as shown in Fig. \ref{fig:gadgets}. 
Let $b \in \{0,1\}$  and let $d_j^{b} \in B$ for some $B \in \hF_i$ and $i \in \{0,\ldots,t-1\}$, then connect $\hd_j^{b}$ and $\bd_j^{b}$,  to vertex $v_B^* \in \bF_*$, connect $\hd_j^{b}$  to $v_B^+ \in \bF_+$, and  connect $\bd_j^{b}$ to $v_B^- \in \bF_-$.
This set of gadgets will be used to select the instance.
\\
$\bullet$ \  Create the gadgets $\Theta_\ell$ for each $\ell \in \{1, \ldots, n\}$ as shown in Fig. \ref{fig:gadgets}.
Connect the vertices $p_\ell,l_\ell$ and $m_\ell$ to $x_\ell^-  \text{ and } x_\ell^+$, and connect the vertex $h_\ell$ to $x_\ell^+$ and the vertex $e_\ell$  to $x_\ell^-$. This gadget will be used to ensure that either $x_\ell^-$ or $x_\ell^+$ is selected. 
\\
$\bullet$ \  Create the gadget $\Gamma$ shown in Fig. \ref{fig:gadgets}. Connect $a$ to each
$\hd^0_j$ and $\hd^1_j$ for $j \in \{0, \ldots, s\}$.
Connect $b$ to each
$\bd^0_j$ and $\bd^1_j$ for $j \in \{0, \ldots, s\}$ and to each $\b_k$ for $k \in \{0, \ldots, r-1\}$.

\begin{example}
    Consider the set splitting instances $(U,\F_0),(U,\F_1)$ where $U = \{x_1,x_2\}$ and $\F_0=\{\{x_1\}\}$ and $\F_1=\{\{x_1,x_2\}\}$. The  families $\hF_0, \hF_1$ are as follow:

\centerline{ 
\begin{tabular}{lll}
$\hF_0=\{B_{0,0}, B_{0,1}\}$  & where
$B_{0,0} =\{x_1, d_1^0, d_0^0 \}$ &  
$B_{0,1} =\{x_1, d_1^1, d_0^1 \}$ 
 \\  
 & & \\
$\hF_1=\{B_{1,0}, B_{1,1}\}$  & where
$B_{1,0} =\{x_1, x_2, d_1^0, d_0^1 \}$ &  
$B_{1,1} =\{x_1, x_2, d_1^1, d_0^0 \}$ 
\\  
 \end{tabular} 
 }
Figure \ref{fig:G} shows the construction of graph $G$ used in the reduction. 
\end{example}
\begin{figure}[t] 
\centering
\includegraphics[width=0.355\linewidth]{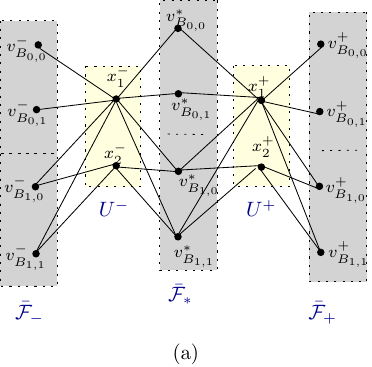} 
\\
\includegraphics[width=0.4\linewidth]{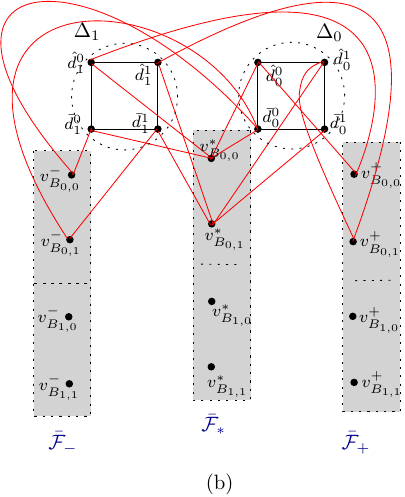} \quad \includegraphics[width=0.45\linewidth]{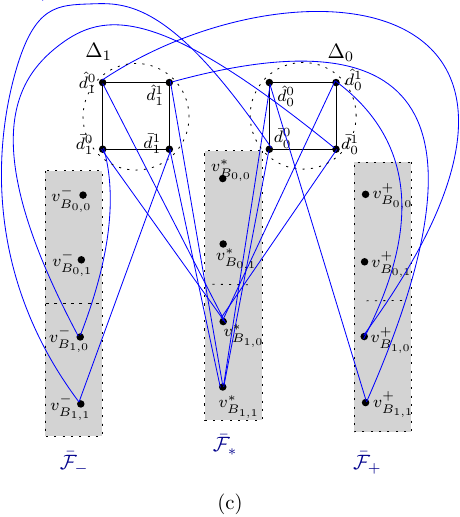}
    \caption{The construction of graph $G$ step by step. To avoid overloading, some edges within the gadgets have been removed. (a) The connection between the vertices in $\bF_+,\bF_-,\bF_*$ and in $\U^-,\U^+$. (b) The connection between the vertices in $\bF_+,\bF_-,\bF_*$ and in $\Delta_0,\Delta_1$}
    \label{fig:G}
\end{figure}

To complete the proof of Theorem \ref{th:no-k}, we prove the correctness of the construction.
\begin{lemma}
    There is an identifying code for $G$ of size at most $6n+\lfloor \log m\rfloor+11 \lceil \log t\rceil + 14=6n+r+11s + 14$ iff there exists at least one $i$ with $i \in \{0, \ldots, t-1\}$  such that $(\F_i,U)$ is a yes instance for Set Splitting.
\end{lemma}
\begin{proof}
    $\Leftarrow$ 
    Let $(\F_i,U)$ be a yes instance for Set Splitting, where  $i=0b_{s-1}\ldots b_{0}$ is the binary representation of $i$, and let $S\subset U$ be the set splitting of  $\F_i$. In such a case, we pick the vertices   $\{\hd_s^{0},\hd_{s-1}^{b_{s-1}},\ldots,\hd_{0}^{b_{0}}\}\cup \{\bd_{s}^{1},\bd_{s-1}^{1-b_{s-1}},\ldots,\bd_{0}^{1-b_{0}}\}$ ( i.e., $2(s+1)$ vertices) and the vertices  $\{x_j^1, x_j^2, x_j^3, x_j^4 \}\cup\{y_j^1, y_j^2, y_j^3, y_j^4\}$ from all the gadgets $\Delta_j$, with $j\in\{0,1,\cdots,s\}$ (i.e.,  $8(s+1)$ vertices). We then pick the vertices $\{x_\ell^+ \in \U^+\ | \ x_\ell \in S\}$  and the vertices $\{x_\ell^- \in \U^-\ | \ x_\ell \in \U \setminus S\}$  (i.e,  $n$ vertices in total). We also pick all the vertices $\{\omega_0,\ldots,\omega_{r-1}\}$ (i.e., $r=\lfloor \log m\rfloor+1$ vertices) and $\{\lambda_0,\ldots,\lambda_s\}$ (i.e., $s+1$ vertices) along with the vertices $a $,$b$ and $c$ (i.e., 3 vertices) from the $\Gamma$ gadget. Finally, we  pick the vertices $p_\ell$,$l_\ell$,$m_\ell,g_\ell,f_\ell$ from each gadget $\Theta_\ell$ for $\ell \in \{1,\ldots, n\}$ (i.e.,  $5n$ vertices). The set of all the picked vertices,  say $Y$, has the cardinality exactly $6n+r+11s+14$. All that remains to claim is the collected vertices form identifying code instance for the constructed graph. First consider the copies of the vertices in $\mathcal{U}$. Copies of the same vertex , in particular $x_\ell^+$ and $x_\ell^-$ have different neighbor in the solution as exactly one is picked and the other is not. Further the vertices in $\mathcal{U}^+\cup \mathcal{U}^- $ are distinguished from all the other vertices in the graph as each vertex $x_\ell^+$ (and $x_\ell^-$) have $p_\ell,l_\ell,m_\ell$ in their neighborhood and no other vertex in the graph have them all at once. Now moving to the vertices in the $\bF_*,\bF_+,\bF_-$, they are all distinguished among themselves using the $\lambda$ vertices (which separates the vertices corresponding to the sets in $\hF_i$ from the vertices corresponding to the sets in $\hF_{i'}$ whenever $i\neq i'$ ) and the $\omega$ vertices (which separates the vertices corresponding to the sets within each $\hF_i$). Further, all of the vertices in $\bF_*\cup \bF_+\cup \bF_-$ their neighbor $\lambda_s$ as well as some combination of vertices from $\mathcal{U}^-\cup \mathcal{U}^+ $ in their solution which  is not true for any other vertex outside the set $\bF_*\cup \bF_+\cup \bF_-$, Finally the copies of the same vertex across $\bF_*, \bF_+$ and $ \bF_-$ are distinguished among themselves due to their interaction with $\mathcal{U}_-$, $\mathcal{U}_+$ and the gadgets $\Delta_j$, in particular $v_B^+$ and $v_B^-$ (when $B\in \hF_i$) are have in their neighborhoods in $\mathcal{U}^+$ and $ \mathcal{U}^- $ the `$+$' copies of the vertices that are contained in $B\cap S$ and the `$-$' copies of the vertices that are contained in $B\cap (\mathcal{U}\setminus S)$ both of which are non empty since $S$ was a splitting set. Here the vertices picked from $\Delta_j$'s differentiate between the vertices $v_{B'}^+,v_{B'}^-$ and $v_{B'}^*$ whenever $B'\notin \hF_i$. For the other gadgets and vertices used, it is straightforward to check that the vertices are distinguished from all the other vertices.

     $\Rightarrow$ 
     Let $S$ be an identifying code of the graph $G$ and $|S|=6n+r+11s + 14$. First note that within each gadget  $\Theta_\ell$, for  $\ell\in \{1,\ldots,n\}$ there must be at least $5$ vertices in $V(\Theta_\ell) \cap S$ to 
     have different closed neighborhood intersection with $S$
     for each pair of  vertices in $V(\Theta_\ell)\setminus \{l_\ell,p_\ell,m_\ell,h_\ell,e_\ell\}$; this give a guarantee that $5n$ vertices from $\bigcup_{i=1}^n V(\Theta_\ell)$ are in $S$. Similarly, at least 3 vertices must be in $V{(D)} \cap S$  to 
     have  different closed neighborhood intersection with $S$
for each pair of
vertices $a', b', c$ in the triangle. Since the vertices  $\{{p}_{\lambda_0}, \ldots,{p}_{\lambda_{s-1}}\} \cup \{{p}_{\omega_0}, \ldots,{p}_{\omega_{r-1}}\}$ have degree 1, then there must be at least $r$ vertices in   
     $(\{{p}_{\omega_0}, \ldots,{p}_{\omega_{r-1}}\}\cup \{\omega_0, \ldots, \omega_{r-1}\}) \cap S$
     and at least $s+1$ vertices in 
     $(\{{p}_{\lambda_0}, \ldots,{p}_{\lambda_{s-1}}\}\cup \{\lambda_0, \ldots, \lambda_{s-1}\}) \cap S$.
Furthermore, since the closed neighborhood of the vertices $h_\ell$ and $e_\ell$ are the same in $\Theta_\ell$, at least one of the vertices in $\{x_{\ell}^-,x_\ell^+\}$ should be in $S$ for all $\ell\in\{1,\ldots,n\}$. This means that at least $n$ vertices from the set $\U^- \cup \U^+$ are in $S$. 
At this point, exactly $10(s+1)$ vertices remain to identify in $S$.
Now we note (as made for $\Theta_\ell$) that if we consider the subgraph induced by the vertices
$\{w_j^i,x_j^i,y_j^i,z_j^i\}$ (where $i\in \{1,2,3,4\}$) in $\Delta_j$, then at least two of these vertices must be in $S$, and then also at least two vertices among $\hd_j^0,\hd_j^1,\bd^0_j,\bd^1_j$ must be in $S$. Hence, 
 at least 10 vertices from each $\Delta_j,\; j\in\{0,1,\cdots,s\}$ should belong to the solution $S$.
In case only 10 vertices from a $\Delta_j$ are in $S$, then necessarily either \textbf{a)} $\{\hd_j^0,\bar{d}^1_j\}\subseteq S$ and $\{\hd_j^1,\bar{d}^0_j\}\cap S=\phi$ or \textbf{b)} $\{\hd_j^1,\bar{d}^0_j\}\subseteq S$ and $\{\hd_j^0,\bar{d}^1_j\}\cap S=\phi$. 
Hence, no other vertex can be in $S$.
    
It remains to show that at least one of the instances of set splitting is a yes instance.
To this aim, we consider (by the above) that either $\{\hd_j^0,\bd^1_j\} \subseteq S$ or $\{\hd_j^1,\bar{d}^0_j\} \subseteq S$  for each $j\in \{0,1,\cdots,s\}$ and we set
$$b_j=\begin{cases}
1 & \mbox{if $\hd^1_j\in S$}\\
0 & \mbox{otherwise.}
\end{cases}$$
We claim that the instance $(\F_i,U)$, where the binary representation of $i$ is
 $i=b_sb_{s-1}\ldots b_0$ if $b_s=0$  and is $i=(1-b_s)(1-b_{s-1})\ldots(1-b_0)$ if $b_s=1$,
 is a yes instance of set splitting with solution set $X=\{x_\ell\in  \U \ | \ x_\ell^+\in S\cap\U^+\}$. \label{here} 

W.l.o.g, assume $b_s=0$. To prove the claim it is sufficient to argue that there is no vertex $v^+_B\in \bF_+$ 
with $B \in \hF_i$ such that  $N[v_B^+]\cap \U^+\subseteq S$ or $(N[v_B^+]\cap \U^+)\cap S=\phi$, that is, the set $A \subset B$ with $A \in \F_i$, is not split by $X$ (recall (\ref{Aadded})).
Assume, on the contrary, that there is such a vertex $v_B^+\in \bF_+$ with $N[v_B^+]\cap \U^+\subseteq S$, then $N[v_B^+]\cap S=N[v_B^*]\cap S$, contradicting our assumption that $S$ is an identifying code for $G$.
Now assume that there is a vertex $v_B^+\in \bF_+$ with $(N[v_B^+]\cap U^+)\cap S=\phi$ then, by construction, $N[v_B^-]\cap U^-\subseteq S$, due to the gadgets $\Theta_\ell$. 
Consider the set  $B'=(B\cup \{d^{1-b_j}_j : 0 \leq j \leq s\})\setminus \{d^{b_j}_j : 0 \leq j \leq s\}\in \hF_i$. In this case $N[v_{B'}^-]\cap S=N[v_{B'}^*]\cap S$, again contradicting our assumption of $S$ being an identifying code for $G$. Therefore, we have that $X$ is a set splitter for the instance $(\F_i,\U)$.
\end{proof}

\begin{theorem}\label{th:no-k-vc}
The \textsc{Identifying Code} problem does not admit a polynomial kernel when parameterized by vertex cover number and solution size unless $NP \subseteq coNP/poly$.
\end{theorem}
\begin{proof}
Along with the theorem proved just now, considering that the subgraph of $G$ induced by the vertices in $\bF_+ \cup \bF_* \cup \bF_-$ is an independent set and  that $|V(G) \setminus (\bF_+ \cup \bF_* \cup \bF_-)|=O(|U|+\log(m) + \log t)$  where  $m=\max_{i: 1\leq i\leq t}|\F_i|$,  we have that $G$ has a vertex cover of size at
most $O(|U|+\log(m) + \log t)$ thus implying the theorem which was left open in the paper \cite{CAPPELLE202168}.
\end{proof}

\bibliographystyle{plain}

\bibliography{ref}
\end{document}